\documentclass[aip,jcp,amssymb,unsortedaddress]{revtex4-1}
\usepackage{amsthm}
\def\Av{{\bf A}}
\def\rv{{\bf r}}
\def\jv{{\bf j}}

\def\0v{{\bf 0}}
\def\Av{{\bf A}}

\newtheorem*{lem}{Lemma}
\begin{document}

\title{On the relation between the scalar and tensor exchange-correlation kernels of the time-dependent density-functional theory}

\author{V.~U.~Nazarov}
\email{nazarov@gate.sinica.edu.tw}

\affiliation{Research Center for Applied Sciences, Academia Sinica, Taipei 115, Taiwan}
\affiliation{Department of Physical Chemistry, Far-Eastern National Technical University, Vladivostok, Russia}

\author{G.~Vignale}
\affiliation {Department of Physics and Astronomy, University of Missouri, Columbia, Missouri 65211, USA}

\author{Y.-C.~Chang}
\affiliation{Research Center for Applied Sciences, Academia Sinica, Taipei 115, Taiwan}

\begin{abstract}
The scalar $f_{xc}$ and tensor $\hat{f}_{xc}$ exchange-correlation (xc) kernels
are key ingredients of the time-dependent density functional theory and
the time-dependent current density functional theory, respectively.
We derive a comparatively simple relation between these two kernels under the assumption that the dynamic  xc
can be considered ``weak". We expect our formula to serve as a convenient bridge between the
scalar $f_{xc}$ which directly enters many applications and
the tensor $\hat{f}_{xc}$ which, due to its locality in space,  is much easier to approximate.
\end{abstract}

\maketitle

\section{Introduction}
The concept of the dynamic xc kernel $f_{xc}$ plays a central role in time-dependent density functional theory
in the linear-response regime.\cite{Gross-85} In the general non-homogeneous case, it is defined as the kernel of the integral transformation
\begin{eqnarray*}
\delta V_{xc}(\rv,\omega)= \int f_{xc}(\rv,\rv',\omega) \delta n(\rv',\omega) d \rv',
\end{eqnarray*}
which relates the change in the dynamic xc potential $\delta V_{xc}$ to the
change in the particle density $\delta n$, where $\omega$ is the frequency. The knowledge of $f_{xc}$ allows one to obtain the density-response
function $\chi$  of interacting electrons through the relation\cite{Gross-85}
\begin{eqnarray*}
\chi^{-1}(\rv,\rv',\omega) \! = \! \chi^{-1}_{KS}(\rv,\rv',\omega) \! - \! \frac{e^2}{|\rv \! - \! \rv'|}\!  - \! f_{xc}(\rv,\rv',\omega),
\end{eqnarray*}
where $\chi_{KS}$ is the density-response function of Kohn-Sham (KS) non-interacting electrons, and $e$ is the charge of the electron.
The scalar $f_{xc}$ is also an important quantity in applications, e.g., it directly enters the formulas for the many-body contribution to the stopping power of materials for slow ions \cite{Nazarov-05} and  the formula for the impurity resistivity of metals.\cite{Nazarov-APS-08}

It is known, however, that at finite frequency $f_{xc}(\rv,\rv',\omega)$ is a strongly nonlocal function of its spacial variables.\cite{Vignale-95}  This non-locality hinders the construction of easy-to-use approximations in time-dependent DFT.
The need to overcome this difficulty has stimulated the development
of the time-dependent {\em current}-density functional theory  (TDCDFT).\cite{Vignale-96}  The key quantity of TDCDFT is the tensor xc kernel $\hat{f}_{xc}$ defined as
\begin{eqnarray*}
\delta A_{i,xc}(\rv,\omega)= \int \hat{f}_{ik,xc}(\rv,\rv',\omega) \, \delta j_{k}(\rv',\omega) \, d \rv',
\end{eqnarray*}
where $\delta \Av_{xc}$ is the change in the xc vector potential and $\delta \jv$ is the change in the current-density.
It has been established that the tensor xc kernel $\hat{f}_{xc}(\rv,\rv',\omega)$ is a much more local function of its space variables than
its scalar counterpart is. As a result, the LDA to TDCDFT has been developed.\cite{Vignale-96,Vignale-97}

It would be natural to expect that any local or semi-local approximation to $\hat{f}_{xc}$  (see, for instance, Ref. \onlinecite{Tao-07}) can be translated to a corresponding approximation for the scalar $f_{xc}$, which is more immediately useful in many applications. The relation between $\hat{f}_{xc}$ and $f_{xc}$ is, however, nontrivial, and it has been established in the general case only recently.  In its exact form it reads\cite{Nazarov-07,Nazarov-08}
\begin{eqnarray}
f_{xc}\! &=& \!
-\frac{e \omega^2}{c}
\nabla^{-2} \nabla \!\cdot\!
\left\{ \! \hat{f}_{xc}
\!+\! \left(\hat{\chi}^{-1}_{KS} \!-\! \hat{f}_{xc}\right)\!
\left[\hat{T} \! \left(\hat{\chi}^{-1}_{KS}\!-\!\hat{f}_{xc}\right)\! \hat{T}\right]^{\!-1}
\right. \cr\cr
& \! \! \times& \left. \! \!  \! \! \left(\hat{\chi}^{-1}_{KS}-\hat{f}_{xc}\right)
\! - \!
\hat{\chi}^{-1}_{KS}
\left(\hat{T}\hat{\chi}^{-1}_{KS}\hat{T}\right)^{-1} \! \!
\hat{\chi}^{-1}_{KS}
\right\}
\cdot \nabla \nabla^{-2}.
\label{via}
\end{eqnarray}
Here  $\hat{\chi}_{KS}$ is  the KS current-density response function (a tensor),
$\hat{T}$ is the projector on the sub-space of transverse vectors, $\nabla$ is the gradient operator, $\nabla^{-2}$ is the inverse Laplacian, and $c$ is the velocity of light.  It must be pointed out that the presence of the KS current-density response function in this expression is unavoidable as long as there is more than one spatial dimension (for the one-dimensional case, see Eq.~(\ref{f0_xc})  below).\cite{Nazarov-07}

Although Eq. (\ref{via}) is exact and general, its direct use is not easy, since it involves repeated inversions
of tensor integral operators. The purpose of this Communication is to propose a simpler and more practical relation even at the price of imposing some restrictions.
As it will be shown below, Eq. (\ref{via}) can be, indeed, considerably simplified by the expansion of the right-hand side to the first order
in $\hat{f}_{xc}$. Physically, this amounts to the intuitively clear assumption that $\hat{f}_{xc}$ is ``small" compared with the KS
contribution, the latter entering Eq. (\ref{via}) through $\hat{\chi}^{-1}_{KS}$.  This assumption also underlies one of the most successful approaches to the calculation of excitation energies, namely the perturbative treatment of $f_{xc}$ in the single-pole approximation of Gross, Dobson, and Petersilka.\cite{Gross-96}   Our final result reads
\begin{equation}
f_{xc} = -\frac{c}{e \omega^2} \, \chi^{-1}_{KS} \nabla \cdot \hat{\chi}_{KS} \cdot \hat{f}_{xc} \cdot \hat{\chi}_{KS}  \cdot \nabla \chi^{-1}_{KS}.
\label{f_xc}
\end{equation}

We find it worthwhile to  put Eq.~(\ref{f_xc}) in words: It says that to apply $f_{xc}$ to a scalar function,  first
a scalar operator $\chi^{-1}_{KS}$ must be applied, then the gradient of the result is taken,
then a tensor operator  $\hat{\chi}_{KS}$ is applied to the vector result of the previous operation,
then a tensor operator  $\hat{f}_{xc}$ and $\hat{\chi}_{KS}$ are consecutively applied, again producing a vector,
then a divergence of this vector is found producing a scalar, and, finally,
a scalar operator $\chi^{-1}_{KS}$ is applied.
We hope that the availability of this relatively simple expression will open the way to the practical use of more sophisticated xc kernels in standard applications of time dependent DFT.

In Sec. \ref{Der}, we provide the derivation of Eq.~(\ref{f_xc}).
In Sec. \ref{Disc}, we discuss the result and draw conclusions.
Explicit expressions for KS response functions, in the form convenient to use with our result,
are collected in the Appendix.

\section{Derivation of Eq.~(\ref{f_xc})}
\label{Der}

To prove Eq. (\ref{f_xc})  we will need the following
\begin{lem}
For the response function $\hat{\chi}$ the relation holds
\begin{eqnarray}
\left(\hat{T} \hat{\chi}^{-1} \hat{T}\right)^{-1}
=
\hat{\chi}+\frac{c}{e \omega^2}
\hat{\chi} \cdot {\bf \nabla} \chi^{-1} {\bf \nabla} \cdot \hat{\chi}.
\label{id}
\end{eqnarray}
The same relation holds for the KS response function $\hat{\chi}_{KS}$.
\end{lem}
\begin{proof}
We can write
\begin{eqnarray*}
\hat{\chi}= \hat{L} \hat{\chi} \hat{L} +  \hat{L} \hat{\chi} \hat{T} + \hat{T} \hat{\chi} \hat{L} +\hat{T} \hat{\chi} \hat{T},
\end{eqnarray*}
where $\hat{L}$ is the longitudinal projector
\begin{equation}
L_{ij}=\nabla^{-2} \nabla_i \nabla_j.
\label{L}
\end{equation}
The same expansion can be performed for the inverse operator $\hat{\chi}^{-1}$. Then
\begin{eqnarray*}
\hat{1}&=&\hat{\chi} \, \hat{\chi}^{-1} \cr\cr
&=& \hat{L} \hat{\chi} \hat{L}  \hat{\chi}^{-1} \hat{L} \! + \!
\hat{L} \hat{\chi} \hat{L}  \hat{\chi}^{-1} \hat{T} \! + \!
\hat{L} \hat{\chi} \hat{T}  \hat{\chi}^{-1} \hat{L} \! + \!
\hat{L} \hat{\chi} \hat{T}  \hat{\chi}^{-1} \hat{T} \!  \! \cr\cr
&+& \hat{T} \hat{\chi} \hat{L}  \hat{\chi}^{-1} \hat{L} \! + \!
\hat{T} \hat{\chi} \hat{L}  \hat{\chi}^{-1} \hat{T} \! + \!
\hat{T} \hat{\chi} \hat{T}  \hat{\chi}^{-1} \hat{L} \! + \!
\hat{T} \hat{\chi} \hat{T}  \hat{\chi}^{-1} \hat{T}.
\label{1}
\end{eqnarray*}
Multiplying this by $\hat{T}$ from the right, we have
\begin{eqnarray*}
\hat{T}=
\hat{L} \hat{\chi} \hat{L}  \hat{\chi}^{-1} \hat{T} \! + \!
\hat{L} \hat{\chi} \hat{T}  \hat{\chi}^{-1} \hat{T} \! + \!
\hat{T} \hat{\chi} \hat{L}  \hat{\chi}^{-1} \hat{T} \! + \!
\hat{T} \hat{\chi} \hat{T}  \hat{\chi}^{-1} \hat{T}.
\label{2}
\end{eqnarray*}
Multiplying the last equality  from the left by $\hat{T}$ and by $\hat{L}$, we will have, respectively,
\begin{eqnarray}
\hat{T}=
\hat{T} \hat{\chi} \hat{L}  \hat{\chi}^{-1} \hat{T} \! + \!
\hat{T} \hat{\chi} \hat{T}  \hat{\chi}^{-1} \hat{T}.
\label{3}
\end{eqnarray}
and
\begin{eqnarray}
\hat{0}=
\hat{L} \hat{\chi} \hat{L}  \hat{\chi}^{-1} \hat{T} \! + \!
\hat{L} \hat{\chi} \hat{T}  \hat{\chi}^{-1} \hat{T}.
\label{4}
\end{eqnarray}
Combining Eqs. (\ref{3}) and (\ref{4}) yields
\begin{eqnarray*}
\left[ \hat{T} \hat{\chi} \hat{T} - \hat{T} \hat{\chi} \hat{L} \left(\hat{L} \hat{\chi} \hat{L}\right)^{-1} \hat{L} \hat{\chi} \hat{T} \right]
\hat{T} \hat{\chi}^{-1} \hat{T}= \hat{T},
\end{eqnarray*}
and therefore
\begin{eqnarray}
\left(\hat{T} \hat{\chi}^{-1} \hat{T}\right)^{-1}=
\hat{T} \left[\hat{\chi}  -  \hat{\chi} \hat{L} \left(\hat{L} \hat{\chi} \hat{L}\right)^{-1} \hat{L} \hat{\chi}\right] \hat{T}.
\label{5}
\end{eqnarray}

The operator in the square brackets in Eq. (\ref{5}) is  purely transverse, as one can verify by applying the $\hat{L}$ operator to its left and to its right and obtaining zero. We can, therefore, drop the $\hat{T}$
operators in the right-hand side of Eq. (\ref{5}), which leads to
\begin{eqnarray}
\left(\hat{T} \hat{\chi}^{-1} \hat{T}\right)^{-1}=
\hat{\chi}  -  \hat{\chi} \hat{L} \left(\hat{L} \hat{\chi} \hat{L}\right)^{-1} \hat{L} \hat{\chi}.
\label{6}
\end{eqnarray}

We now recall the relation between the scalar and tensor response functions
\begin{eqnarray}
\chi= -\frac{c}{e \, \omega^2}
\nabla \cdot \hat{\chi} \cdot \nabla,
\label{chit}
\end{eqnarray}
which, combined with Eqs. (\ref{6}) and (\ref{L}), yields Eq. (\ref{id}).
\end{proof}

Now we expand the right-hand side of Eq. (\ref{via}) to the first order in $\hat{f}_{xc}$:
\begin{eqnarray*}
&& f_{xc} \! =
\! - \! \frac{e \omega^2}{c} \nabla^{-2} \nabla \! \cdot \!  \left\{ \! \hat{f}_{xc}
\!-\! \hat{f}_{xc}
\! \left(\hat{T} \hat{\chi}^{-1}_{KS}\! \hat{T}\right)^{\!-1} \! \! \! \! \hat{\chi}^{-1}_{KS}
\! - \! \hat{\chi}^{-1}_{KS}  \left(\hat{T}\hat{\chi}^{-1}_{KS}\hat{T}\right)^{\! -1}  \right. \cr\cr
&&\left.
\times
\hat{f}_{xc} + \hat{\chi}^{-1}_{KS}
\left(\hat{T}\hat{\chi}^{-1}_{KS}\hat{T}\right)^{-1} \hat{f}_{xc}\left(\hat{T}\hat{\chi}^{-1}_{KS}\hat{T}\right)^{\! -1} \! \hat{\chi}^{-1}_{KS}
\right\}
\! \cdot \! \nabla \nabla^{-2}.
\label{via1}
\end{eqnarray*}
Equation (\ref{f_xc}) immediately follows from the above equation by use of Eq.~(\ref{id}).

\section{Discussion and conclusions}
\label{Disc}

We point out that similar to its full form of Eq. (\ref{via}),
in the simplified form of Eq. (\ref{f_xc}), the scalar $f_{xc}$ depends not on
the tensor $\hat{f}_{xc}$ only, but also on the KS response function.
However, in the case of {\em purely longitudinal} $\hat{f}_{xc}$ (e.g., for 1D inhomogeneity)  Eq.~(\ref{f_xc}) reduces to
\begin{equation}
\hat{f}_{xc} = -\frac{c}{e\omega^2} \nabla f_{xc}   \nabla.
\label{f0_xc}
\end{equation}
This agrees with the expression suggested in Ref.~\onlinecite{Dion-05}. We, however, note that in the inhomogeneous 2D and 3D cases, there is no reason for $\hat{f}_{xc}$ to be purely longitudinal (and it, indeed, explicitly is not such within LDA\cite{Vignale-97}). 

It is known  that at isolated frequencies the operator
$\chi_{KS}$ can have zero eigenvalues and, therefore, be non-invertible.\cite{Mearns-87,Gross-88}
This poses an interesting question whether or not the scalar $f_{xc}$  can have a singularity due to the presence of $\chi^{-1}_{KS}$ in Eq.~(\ref{f_xc}) even if  the tensor $\hat{f}_{xc}$ is non-singular. This evidently does not happen in 1D case,
when the presence of $\hat{\chi}_{KS}$ in Eq.~(\ref{f_xc})  compensates the possible singularity leading to Eq.~(\ref{f0_xc}).  However, in the general case, we cannot rule out the possibility that a scalar potential that produces no density response in the KS system (i.e., a null eigenvector of $\chi_{KS}$) may nevertheless produce a finite transverse current response in the same system.  If this happens, then the singularity in $\chi^{-1}_{KS}$ remains uncompensated, and $\hat f_{xc}$ may have singularities that are not present in $f_{xc}$. 
This kind of singularity is, however, impossible
at complex frequencies with finite imaginary part, since in this case
the density-response function $\chi_{KS}$ is invertable.\cite{Ng-89}

In conclusion,
we have considerably simplified the relation between the two key quantities of the time-dependent density
functional theory and the time-dependent current density functional theory: the scalar $f_{xc}$ and the tensor $\hat f_{xc}$, respectively.
This has been achieved at the price of assuming the dynamic exchange-correlations to be weak on the background of the Kohn-Sham response, and solving the problem to the first order with respect to the former. 
We hope that the availability of this approximation, similar in spirit to the perturbative approximation in the Gross-Dobson-Petersilka approach, will stimulate the use of more accurate xc kernels in standard applications of time-dependent DFT. 

\section{Acknowledgements}
This work was supported by Academia Sinica and  by DOE grant DE-FG02-05ER46203. VUN gratefully acknowledges the hospitality of the Department of Physics and Astronomy, Universityof Missouri, Columbia.

\appendix*
\section{}
Apart from the tensor xc kernel $\hat{f}_{xc}$, which is considered an input quantity within the context of this paper,
to use Eq.~(\ref{f_xc}) one needs the tensor
$\hat{\chi}_{KS}$ and the scalar $\chi_{KS}$.
While the explicit forms of the latter two operators are well known, the purpose of this appendix
is to conveniently represent the construct $\hat{\chi}_{KS}\cdot \nabla$, which enters Eq.~(\ref{f_xc}).
The tensor KS response functions can be written  as
\begin{widetext}
\begin{eqnarray}
&&\hat{\chi}_{KS,ij}({\bf r},{\bf r}',\omega) =
\frac{e}{ c m} \, n_0(r) \delta({\bf r}-{\bf r}') \, \delta_{ij} -\frac{e}{4 c m^2} \times \cr\cr
&& \sum\limits_{\alpha \beta} \! \! \frac{f_\alpha-f_\beta}{\omega \! - \! \epsilon_\beta \! + \! \epsilon_\alpha \! + \! i\eta}
\left[\psi^*_\alpha({\bf r}) \nabla_i \psi_\beta({\bf r}) \! - \! \psi_\beta({\bf r}) \nabla_i \psi^*_\alpha({\bf r})\right]
\left[ \psi^*_\beta({\bf r}') \nabla'_j \psi_\alpha({\bf r}') \! - \! \psi_\alpha({\bf r}') \nabla'_j \psi^*_\beta({\bf r}') \right], \ \ \ \
\label{chi_KS}
\end{eqnarray}
\end{widetext}
where $\psi_\alpha({\bf r})$ and $\epsilon_\alpha$ are KS  wave-function
and eigenenergy, respectively, in the state $\alpha$, $f_\alpha$ is the occupation number of this state,
and $\eta$ is an infinitesimal positive. Applying $\nabla$ operator to Eq.~(\ref{chi_KS}) from its right,
gives us a convenient form of the operator $\hat{\chi}_{KS} \cdot \nabla$ which enters Eq.~(\ref{f_xc}) 
\begin{eqnarray}
&&\hat{\chi}_{KS,ij}({\bf r},{\bf r}',\omega) \nabla'_j  =
\frac{e \omega}{2 c m} \sum\limits_{\alpha \beta} \frac{f_\alpha-f_\beta}{\omega-\epsilon_\beta+\epsilon_\alpha+i\eta}\cr\cr
&&\times \left[\psi^*_\alpha({\bf r}) \nabla_i \psi_\beta({\bf r})- \psi_\beta({\bf r}) \nabla_i \psi^*_\alpha({\bf r})\right]
\psi^*_\beta(\rv') \psi_\alpha(\rv'). \  \  \ \ \ \ \
\label{chirightsc1}
\end{eqnarray}
The fact that $\psi_\alpha(\rv)$ satisfy  \protect{Schr\"{o}dinger's} equation with the eigenvalues $\epsilon_\alpha$ has been used.
It can be easily verified that  applying $\nabla$ to Eq.~(\ref{chirightsc1}) from its left and using Eq.~(\ref{chit}) leads to a known
expression for the KS density-response function
\begin{eqnarray*}
\chi_{KS}(\rv,\rv' \! ,\omega)  \! =\! \!
\sum\limits_{\alpha \beta} \!  \frac{f_\alpha-f_\beta}{\omega \! - \! \epsilon_\beta \! + \! \epsilon_\alpha \! + \! i\eta}
 \psi^*_\alpha(\rv) \psi_\beta(\rv)  \psi^*_\beta(\rv') \psi_\alpha(\rv').  \cr\cr
\ \label{chi_KS_s}
\end{eqnarray*}

\end{document}